\documentclass[reprint,aps,pra,a4paper,superscriptaddress,longbibliography,preprintnumbers,nofootinbib,floatfix]{revtex4-2}
\usepackage[utf8]{inputenc}
\usepackage[english]{babel}
\usepackage[T1]{fontenc}
\usepackage{amsmath}
\usepackage{amsthm}
\usepackage{amsfonts}
\usepackage{xcolor}
\usepackage{graphicx}
\usepackage{hyperref}
\usepackage{braket}
\usepackage{mathtools}
\definecolor{mpBlue}{RGB}{21, 101, 192} 
\definecolor{mpRed}{RGB}{198, 40, 40}   
\definecolor{mpGreen}{RGB}{46, 125, 50} 
\hypersetup{citecolor=mpGreen}
\hypersetup{colorlinks=true}
\hypersetup{linkcolor=mpBlue}
\hypersetup{urlcolor=mpBlue}
\DeclareMathOperator{\e}{e}
\let\Re\relax
\DeclareMathOperator\Re{Re}
\let\Im\relax
\DeclareMathOperator\Im{Im}
\DeclareMathOperator\RR{\mathbb{R}}
\DeclareMathOperator\PP{\mathbb{P}}
\newcommand*\dd{\mathrm{d}}
\DeclarePairedDelimiter\abs\lvert\rvert

\newtheorem{lemma}{Lemma}

\begin{document}
\title{Operational theory for photonic circuits: generalizing linear optics beyond quantum theory}

\author{Ismaël Septembre}
\email{ismael.septembre@uni-siegen.de}
\affiliation{Naturwissenschaftlich-Technische Fakultät, Universität Siegen, 57068 Siegen, Germany}
\author{Matthias Kleinmann}
\email{matthias.kleinmann@uni-muenster.de}
\affiliation{Department for Quantum Technology, Universität Münster, Heisenbergstraße 11, 48149 Münster, Germany}
\affiliation{Naturwissenschaftlich-Technische Fakultät, Universität Siegen, 57068 Siegen, Germany}
\author{Martin Plávala}
\email{martin.plavala@uni-hannover.de}
\affiliation{Institut für Theoretische Physik, Leibniz Universität Hannover, 30167 Hannover, Germany}

\begin{abstract}
We use the analogy between the quadrature operators in quantum optics and the phase-space coordinates in a mechanical harmonic oscillator to generalise the theory of linear optics beyond quantum theory. For this, we introduce a framework for analysing photonic circuits within generalised probabilistic theories based on quadrature operators. Using this framework, we ask whether phenomena like the Hong-Ou-Mandel effect, Mach-Zehnder interferometry are specific to quantum theory or also occur across a broader class of theories. We find the latter: the vanishing photon coincidences at the output of a beam splitter also occur for classical light and for a modified version of quantised light. Similarly, for Mach-Zehnder interferometry we find that there is no qualitative difference between quantum theory and these alternative theories. Our framework can serve as a foundation to study generalised theories of complex photonic set-ups and photonic quantum computers.
\end{abstract}

\maketitle


\section{Introduction}
Genuine quantum effects are at the heart of the development of quantum technologies which offer, for example, advantages in quantum communication and quantum sensing and promise speed-ups in quantum computing. Hence, it is crucial to understand which effects are genuinely quantum and are then an essential resource enabling quantum advantages. There is an ongoing effort to distinguish genuine quantum effects from ones that could be observed classically \cite{catani2023interference, t2024hidden, wetterich2025quantum, hance2026noncontextual}. We focus here on quantum effects that occur in photonic set-ups, in particular in interferometric set-ups using only linear optical elements, namely mirrors, phase shifters and beam splitters \cite{knill2001scheme} (in addition to single-photon sources, detectors, \emph{etc}). Such set-ups include, for example, the photonic boson sampler that has been used to demonstrate quantum computational advantage \cite{aaronson2011computational, zhong2020quantum}. In contrast, classical optics is not sufficient to observe other effects that are already established to be genuinely quantum, in particular the violation of Bell inequalities \cite{aspect1981experimental, aspect1982experimental, aspect1982experimental2, brunner2014bell}. A key example of a quantum effect in linear optics is the Hong-Ou-Mandel effect \cite{hong1987measurement, shih1988new}, but it has recently been demonstrated that the vanishing of the coincidence counts can actually be observed with classical states \cite{sadana2019near, na2020classical}. Other canonical set-ups are interferometers, which play a key role in quantum sensing \cite{Degen_2017} and, for example, interaction-free measurements \cite{Elitzur_1993, Kwiat_1995}.

In order to get a clearer understanding of the quantum nature of such effects, it can be useful to step back and study a more general framework in which quantum theory appears as a special case. This is the idea of generalised probabilistic theories which include classical and quantum theories \cite{plavala2023general}. Recently, this framework was extended to include continuous-variable theories such as the harmonic oscillator \cite{plavala2022operational} and the hydrogen atom \cite{plavala2023generalized}.

In this article, we take advantage of this formulation to propose a model for studying photonic circuits in generalised probabilistic theories, in particular by modelling beam splitters and phase shifters. We then consider the Hong-Ou-Mandel and the Mach-Zehnder set-up in classical and quantum optics, as well as in alternative scenarios. We show that the vanishing of the coincidence probability (Hong-Ou-Mandel coincidence dip) on a beam splitter is a general feature of this type of generalised optics, see Figs.~\ref{fig:quantum-coincidence}, \ref{fig:classical-coincidence}, and \ref{fig:sawtooth-coincidence}. Similarly the measurement signal in Mach-Zehnder interferometers does not show qualitative different behaviours, see Figs.~\ref{fig:classical-interference}, \ref{fig:quantum-interference}, and \ref{fig:sawtooth-interference}. Finally we consider the achievable correlations in an interferometric setup, see Fig~\ref{fig:effective-state-spaces}, and different theories have characteristic sets of correlations, however with the set of correlations in classical optics being the most general case, including the correlations of the other theories.


\section{Quantum optical phase space}
A beam splitter is an optical element that splits an incident beam of light into two parts, one reflected and one transmitted. In classical optics, a beam splitter is interpreted as dividing the intensity of the incoming light into two directions. In contrast, a single photon entering a beam splitter exits in a superposition of the states ``photon exiting in output $C$'' and ``photon exiting in output $D$'', see Fig.~\ref{figBeamSplitter}. Thus, quantum and classical optics are fundamentally different in their interpretation and predictions. However, they can be compared because the action of a beam splitter can be represented as the same transfer matrix in both cases \cite{grynberg2010introduction}. In classical optics, electromagnetic fields are conveniently represented using complex amplitudes, whereas in quantum optics, the fields are represented by operators. The correspondence between the two arises because classical field amplitudes are identified with expectation values of quantum quadrature operators.

For the classical case, the beam splitter is described by a unitary matrix relating the input and output electric field amplitudes,
\begin{equation}\label{bsClassical}
    \begin{pmatrix}E_C\\E_D\end{pmatrix} = U_\mathrm{bs}\begin{pmatrix}E_A\\E_B\end{pmatrix}.
\end{equation}
In quantum optics, the quantisation of the electromagnetic fields promotes the classical mode amplitudes to bosonic field operators. For the same unitary transformation, the input--output relations then take an same form,
\begin{equation}\label{bsQuantum}
    \begin{pmatrix}\hat a^\dag_C\\\hat a^\dag_D\end{pmatrix}
    = U_\mathrm{bs} \begin{pmatrix}\hat a^\dag_A\\\hat a^\dag_B\end{pmatrix}.
\end{equation}
This correspondence becomes particularly transparent when expressing the field operators in terms of quadratures:
\begin{equation}\begin{split}\label{eqxpU}
    \vec x^\mathrm{\,(out)} &= \Re(U) \vec x^\mathrm{\,(in)} + \frac{x_0}{p_0} \Im(U) \vec p^\mathrm{\,(in)}, \\
    \vec p^\mathrm{\,(out)} &= -\frac{p_0}{x_0} \Im(U) \vec x^\mathrm{\,(in)} + \Re(U) \vec p^\mathrm{\,(in)},
\end{split}\end{equation}
with $\vec x^\mathrm{\,(in)}=(\hat x_A,\hat x_B)$, $\vec x^\mathrm{\,(out)}=(\hat x_C, \hat x_D)$ and analogously for $\vec p$.

In quantum optics, quadratures are usually dimensionless and hence $x_0=1=p_0$. In order to obtain a clean  phase-space picture (and also to avoid ambiguities in the definition of the quadrature operators), we use the analogy between the electromagnetic field and a mechanical harmonic oscillator with mass $m$ and angular frequency $\omega$, that is, we use the Hamiltonian
\begin{equation}\label{eq:hamilton-base}
    H(r) = \frac{p^2}{2m}+ \frac12m\omega x^2
\end{equation}
where now $x_0=\sqrt{\hbar/m\omega}$ and $p_0=\sqrt{\hbar m\omega}$. We will use this analogy to infer operational properties of optical elements in general operational theories in phase space, beyond classical and quantum optics.

\begin{figure}
    \centering
    \includegraphics[width=0.7\linewidth]{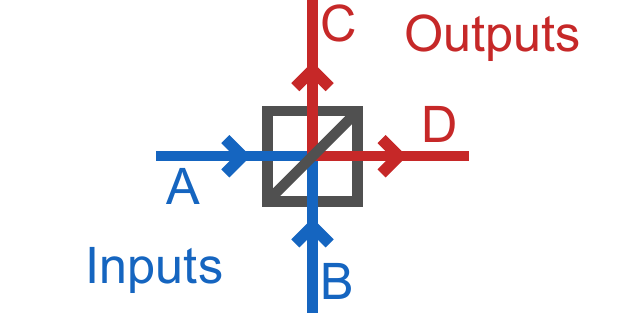}
    \caption{Sketch of an optical beam splitter with input ports $A$, $B$ and output ports $C$, $D$. The electric field amplitudes are transformed according to Eq.~\eqref{bsClassical} and the relation of the mode creation operators is given in Eq.~\eqref{bsQuantum}.}
    \label{figBeamSplitter}
\end{figure}


\section{Generalised probabilistic theories in phase space}
Equation~\eqref{eqxpU} links the quadratures of the input and output modes through the coefficients of the transfer matrix. However, it does not yet tell us anything about the underlying theory. In particular, at this level, classical and quantum optics would seem equivalent. The difference between the theories becomes apparent once we consider the set of admissible states of the theory and the description of the observables and measurements. In this section we review the formulation of harmonic oscillators in phase space presented in Ref.~\cite{plavala2022operational} that precisely captures these differences between classical, quantum, and other operational theories.

We suppose that the states described by the theory are given by a pseudo-probability distribution $\rho(x,p)$ on the phase space $(x,p)$, where $\rho$ is real-valued, but can take negative values and is normalised as $\int_{\RR^2} \rho(x,p) \dd{x} \dd{p} = 1$. Observables of the theory are given by real-valued functions $A(x,p)$ in phase space. We write $\tilde{A}$ for the random variable corresponding to the outcomes of the observable $A$. For a system in state $\rho$ the mean value of the observable $A$ is then given as
\begin{equation}
    \braket{\tilde{A}}_\rho  = \int_{\RR^2} A(x,p) \rho(x,p) \dd x \dd p = \braket{A, \rho},
\end{equation}
where we have introduced the shorthand notation for the integral. The probability that $\tilde{A}$ assumes a value in an interval $I$ is
\begin{equation}
    \PP_\rho (\tilde{A} \in I ) = \int_{\RR^2} g_A(I;x,p) \rho(x,p) \dd x \dd p = \braket{g_A(I), \rho},
\end{equation}
where $g_A(I;x,p)$ is the phase space spectral measure of the operator $A$ which must fulfil \cite{plavala2022operational}
\begin{align}
    \int_{\RR} g_A(a; x,p) \dd{a} &= 1, \\
    \int_{\RR} a g_A(a; x,p) \dd{a} &= A(x,p).
\end{align}
We see that having access to a state and a spectral measure in phase space allows calculating probabilities. In classical physics, the spectral measure is piece-wise constant for any fixed interval $I$, see also Eq.~\eqref{eq:classicalspecfu} below, while there is no intrinsic uncertainty in preparing the states, that is, states can be represented by mixtures of $\delta$-distributions in phase space. In contrast, in quantum theory, in phase space the states are represented by Wigner functions \cite{case2008wigner} that can attain negative values, foreshadowing quantum effects.

\subsection{Harmonic oscillator}
It is useful to consider the spectral decomposition of the Hamiltonian in phase space. The Hamiltonian of a mechanical oscillator, see Eq.~\eqref{eq:hamilton-base}, can be written as
\begin{equation}
    H(r) =\frac12 h_0 r^2 \quad\text{with }
    r=\sqrt{\frac{p^2}{p_0^2} + \frac{x^2}{x_0^2}},
\end{equation}
where $h_0=x_0p_0 \omega$ and $x_0$, $p_0$ as before.
If we assume that the energy of the harmonic oscillator has discrete levels $E_n$, $n=0,1,\dotsc$ and that the phase space spectral measure of the Hamiltonian only depends on the dimensionless phase-space radius $r$, then
\begin{equation}
    g_H(I;x,p) = \sum_{n\colon E_n\in I} T_n (r^2),
\end{equation}
where the decomposition has to fulfil \cite{plavala2022operational}
\begin{equation}
    \sum_{n=0}^\infty T_n(r^2)=1 \quad \text{and}
    \quad \sum_{n=0}^\infty E_n T_n(r^2)=H(r).
\end{equation}
The allowed values that the energy observables can reach are given by the spectral values $E_n$. It remains to specify the states $\rho$ that are available in the theory.

A state $\rho$ is a sharp energy state if $\PP_\rho(\tilde{E} = E) = 1$ holds, that is, if $E$ is the only possible outcome of the energy measurement of state $\rho$. In quantum theory, sharp energy states correspond to energy eigenstates (for pure states). But equivalently one can also define these eigenstates as the time independent pure states. This is not the case in other operational theories, for example already in classical theory there are time dependent sharp energy states.

The time evolution can be introduced by generalizing the Moyal bracket \cite{plavala2023general}. Here we only need to consider the case of harmonic Hamiltonians. Let $\Phi_t$ denote the super-operator that shifts the system in time by $t$, that is, $\Phi_t(\rho)$ is the state $\rho$ shifted forward by time $t$. Therefore, the time evolution in a harmonic potential is a rotation in phase space (see details in Appendix~\ref{appA}):
\begin{equation}\label{eq:timevo}
    \Phi_t(\rho)(r,\theta) = \rho(r,\theta + \omega t),
\end{equation}
where the state $\rho(x,p)$ is written in polar coordinates $\rho(r,\theta)$ with $x=x_0 r \cos(\theta)$ and $p=p_0 r \sin(\theta)$. Thus, if a state has no angular dependence, it does not evolve in time. This is the case for energy eigenstates in quantum theory because Wigner functions of the harmonic oscillator are invariant under rotation around the origin. In general, this is not the case, and two sharp energy states associated with the same energy $E_n$ can differ in their individual phase. Finally, the evolution of a state passing through a phase shifter can be expressed as
\begin{equation}\label{eq:phase}
    \Phi_\varphi (\rho)(r,\theta) = \rho(r,\theta + \varphi),
\end{equation}
where we used that a phase-shifter induces $a^\dag\mapsto \e^{i \varphi}a^\dag$ or, correspondingly for the quadratures,
\begin{equation}\begin{split}
    x^\mathrm{\,(out)}&=\cos(\varphi)x^\mathrm{\,(in)} + \frac{x_0}{p_0} \sin(\varphi)p^\mathrm{\,(in)},\\
    p^\mathrm{\,(out)}&=-\frac{p_0}{x_0}\sin(\varphi)x^\mathrm{\,(in)} + \cos(\varphi)p^\mathrm{\,(in)}.
\end{split}\end{equation}


\subsection{Beam splitters}
In bipartite configurations where two light modes are considered, as is the case in the Hong-Ou-Mandel effect, a joint state can be constructed as a product of two modes. To be concrete, let $\rho_1(r_1,\theta_1)$ and $\rho_2(r_2,\theta_2)$ be states in polar coordinates of two harmonic oscillators. Then their joint state is given by:
\begin{equation}
    \rho(r_1,\theta_1,r_2,\theta_2) = \rho_1(r_1,\theta_1) \rho_2(r_2,\theta_2).
\end{equation}
It is now convenient to replace $\theta_1$ and $\theta_2$ with relative and absolute phases. We define the absolute phase as $\theta_\mathrm{abs} = (\theta_1 + \theta_2) / 2$ and the relative phase as $\theta_\mathrm{rel} = \theta_2 - \theta_1$. Then, assuming both oscillators have the same angular frequency, we get that the time evolution only shifts the absolute phase but leaves the relative phase invariant, that is,
\begin{multline}
    (\Phi_t \otimes \Phi_t)(\rho)(r_1,r_2,\theta_\mathrm{abs},\theta_\mathrm{rel}) =\\
    \rho(r_1,r_2,\theta_\mathrm{abs} +\omega t,\theta_\mathrm{rel}).
\end{multline}

The action of a beam splitter on a bipartite state can be derived from Eq.~\eqref{eqxpU}: let $\Phi_U$ be the linear map on states representing the beam splitter corresponding to a unitary matrix $U$. Then the action of $\Phi_U$ on a bipartite state $\rho$ is best expressed in Cartesian coordinates:
\begin{multline} \label{eq:beamSplitterRho}
    \Phi_U (\rho)(\vec{x}, \vec{p}) =
    \rho(\Re(U)^T \vec{x} - \frac{x_0}{p_0} \Im(U)^T \vec{p}, \\\frac{p_0}{x_0} \Im(U)^T \vec{x} + \Re(U)^T \vec{p}).
\end{multline}
We will later use the following result that shows that beam splitters commute with the time evolution if angular frequencies of the input and output modes are the same.
\begin{lemma} \label{lemma:BScommT}
Consider a beam splitter described by a unitary matrix $U$ and inducing the transformation $\Phi_U$ as in Eq.~\eqref{eq:beamSplitterRho}. Then a phase shift as in Eq.~\eqref{eq:phase} on both input ports or both output ports commutes with $\Phi_U$, that is, $\Phi_U\circ (\Phi_\varphi\otimes \Phi_\varphi) = (\Phi_\varphi\otimes \Phi_\varphi) \circ \Phi_U$.
\end{lemma}
\begin{proof}
The phase-space transformation in Eq.~\eqref{eqxpU} can be written as $(\vec x,\vec p)\mapsto R_U(\vec x,\vec p)$ with $R_U= \openone \otimes \Re(U) + S\otimes \Im (U)$, where $\otimes$ is the Kronecker product and
\begin{equation}
    S= \begin{pmatrix}0&\frac{x_0}{p_0}\\-\frac{p_0}{x_0}&0\end{pmatrix}.
\end{equation}
A phase shift $\theta\mapsto \theta+\varphi$ of one mode induces the phase-space transformation $R_\varphi=\cos(\varphi)\openone+ \sin(\varphi) S$ and hence, for two modes the phase-space transformation is $(\vec x,\vec p)\mapsto (R_\varphi\otimes \openone)(\vec x,\vec p)$. The claim follows directly from $[R_U, R_\varphi\otimes \openone]=0$.
\end{proof}


\section{Hong-Ou-Mandel effect in generalised probabilistic theories}
In a typical Hong-Ou-Mandel experiment as depicted in Fig.~\ref{figHOM}, two photons are sent into the distinct input ports of a beam splitter simultaneously. In quantum optics, the photons are then found to exit the beam splitter in the same output mode. Clearly the whole effect is centred around each input mode containing one photon but one output mode containing no photons. We capture this more exactly via the coincidence probability.

\begin{figure}
    \centering
    \includegraphics[width=0.7\linewidth]{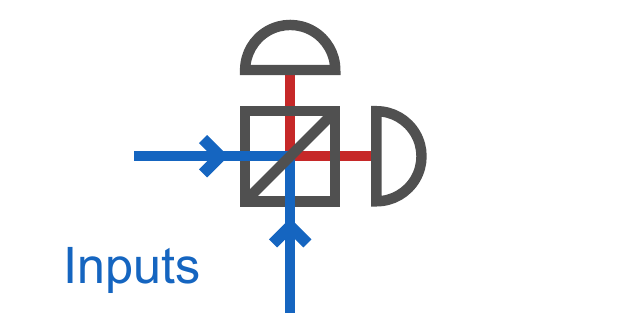}
    \caption{Sketch of a Hong-Ou-Mandel set-up consisting of two input ports, beam splitter, and coincidence measurement on the outputs of the beam splitter.}
    \label{figHOM}
\end{figure}

Given a state $\rho_{AB}$ of two modes (or, equivalently, two harmonic oscillators) labelled as $A$ and $B$, the coincidence probability is the probability of measuring both $A$ and $B$ to have energy greater than some energy threshold $E_\mathrm{th}$, that is:
\begin{equation}\label{probaco}\begin{split}
    P_\mathrm{co}(\rho_{AB})
    &= \PP_{\rho_{AB}} (\tilde E^{(A)} > E_\mathrm{th} \text{ and } \tilde E^{(B)} > E_\mathrm{th}) \\
    &= \braket{g_H^{(A)}((E_\mathrm{th}, \infty)) g_H^{(B)}((E_\mathrm{th}, \infty)), \rho_{AB}}.
\end{split}\end{equation}
In quantum theory, one would choose $E_\mathrm{th}$ to be between the vacuum and single photon energy. Then the coincidence probability is the probability of observing at least one photon in each mode. In theories without quantised energy levels (and hence without photons) one may choose $E_\mathrm{th}$ to be an arbitrary threshold, for example for detectors based on the photo-effect, where an energy gap has to be passed.

We can consider the same experiment in an arbitrary generalised probabilistic theory in phase space. We start with a state $\rho_{AB}$ such that $P_\mathrm{co}(\rho_{AB}) = 1$ and we want to investigate the coincidence probability of the output of the beam splitter $P_\mathrm{co}(\Phi_U (\rho_{AB}))$. For this we assume that the energy phase-space spectral measure does not depend on the phase $\theta$, $g_H(I,r,\theta)=g_H(I,r)$. This is certainly the case in quantum and classical theory. Then, due to
\begin{multline}
    \braket{g_H(I),\rho}=
    \int g_H(I,r) \rho(r,\theta) \,x_0p_0r \,\dd r \,\dd\theta=\\
    \int g_H(I,r) \left(\int \rho(r,\theta) \dd\theta \right)x_0p_0r \dd r,
\end{multline}
the coincidence probability does not depend on the phases of the output states and hence, by virtue of Lemma~\ref{lemma:BScommT}, it also does not depend on the global phase of the input states. However, the final coincidence probability depends on the relative phase between the two input states. For this reason, we introduce the visibility $V$, computed from the average of the coincidence probability over the relative phase:
\begin{equation}\label{eqvisibility}
    V = 1 - \frac{1}{2\pi}\int_{0}^{2\pi} P_\mathrm{co}(\Phi_U(\rho_{AB}(\theta_\mathrm{rel}))) \dd\theta_\mathrm{rel},
\end{equation}
where $\rho_{AB}(\theta_\mathrm{rel})$ is the bipartite state with relative phase $\theta_\mathrm{rel}$.

\begin{figure}
    \centering
    \includegraphics[width=0.7\linewidth]{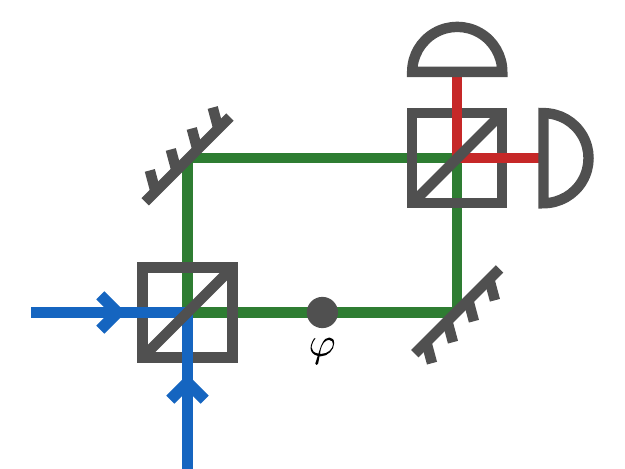}
    \caption{Sketch of a Mach-Zehnder interferometer. The two input modes first interfere on a beam splitter, then the phase of one of the arms is shifted by $\varphi$, then the modes again enter a beam splitter and a coincidence measurement is performed on the output.}
    \label{figHOM_interferometer}
\end{figure}

It was also suggested in Ref.~\cite{sadana2019near} that a Mach-Zehnder interferometric set-up could provide complementary information that is not directly available through standard Hong-Ou-Mandel set-up. The complete set-up is depicted in Fig.~\ref{figHOM_interferometer}: the two input modes interfere on the first beam splitter, then the phase of one of the arms is shifted by an angle $\varphi$ using a phase shifter, and after interfering on another beam splitter the coincidence is measured again. The measurement signal $S(\varphi)$ of such an experiment is a coincidence probability that depends on the angle $\varphi$. Below we also describe this set-up in phase space for arbitrary harmonic oscillator modelling the two modes.

In the following, we take some examples of generalised probabilistic theories in phase space and compute the visibility of the Hong-Ou-Mandel coincidence dip and the Mach-Zehnder set-up. Throughout we use the balanced beam splitter with
\begin{equation}
    U = \frac1{\sqrt2}\begin{pmatrix}1&1\\1 & -1\end{pmatrix}.
\end{equation}

\subsection{Quantum theory}
First, we consider the case of quantum theory. In quantum optics, it is well known that the visibility of the Hong-Ou-Mandel coincidence dip is 100\% \cite{bouchard2020two}. Our framework yields the same result, as noted. The initial state $\ket{11}$ consists of one photon in each input port, so $P_\mathrm{co}(\rho_{AB}) = 1$. After the beam splitter, the state becomes, by definition, $\rho_{CD} \equiv \Phi_U (\rho_{AB})$. Standard quantum optics computation shows that this state is $(\ket{20} - \ket{02}) / \sqrt{2}$, a superposition of both photons in output $C$ and both photons in output $D$. In other words, there is no possibility for the photons to come out in different output ports simultaneously: $P_\mathrm{co}(\Phi_U (\rho_{AB})) = 0$. Since in quantum theory energy eigenstates are time independent, the state does not depend on the relative phase and thus we get $V = 100\%$. This is usually referred to as the Hong-Ou-Mandel coincidence dip.

In the Mach-Zehnder interferometer, the state after the phase shifter is $(\ket{20} - \e^{2i\varphi} \ket{02}) / \sqrt{2}$, which is transformed by the second beam splitter to the final state $\cos(\varphi) \ket{11} - i \sin(\varphi) (\ket{20} + \ket{02}) / \sqrt{2}$. The coincidence measurement then gives the measurement signal
\begin{equation} \label{eq:QT-measurement}
    S(\varphi) = \cos(\varphi)^2
\end{equation}
that represents the interference on the beam splitters, in other words, the coincidence probability at the output of the entire interferometer.

\begin{figure}
    \centering
    \includegraphics[width=\linewidth]{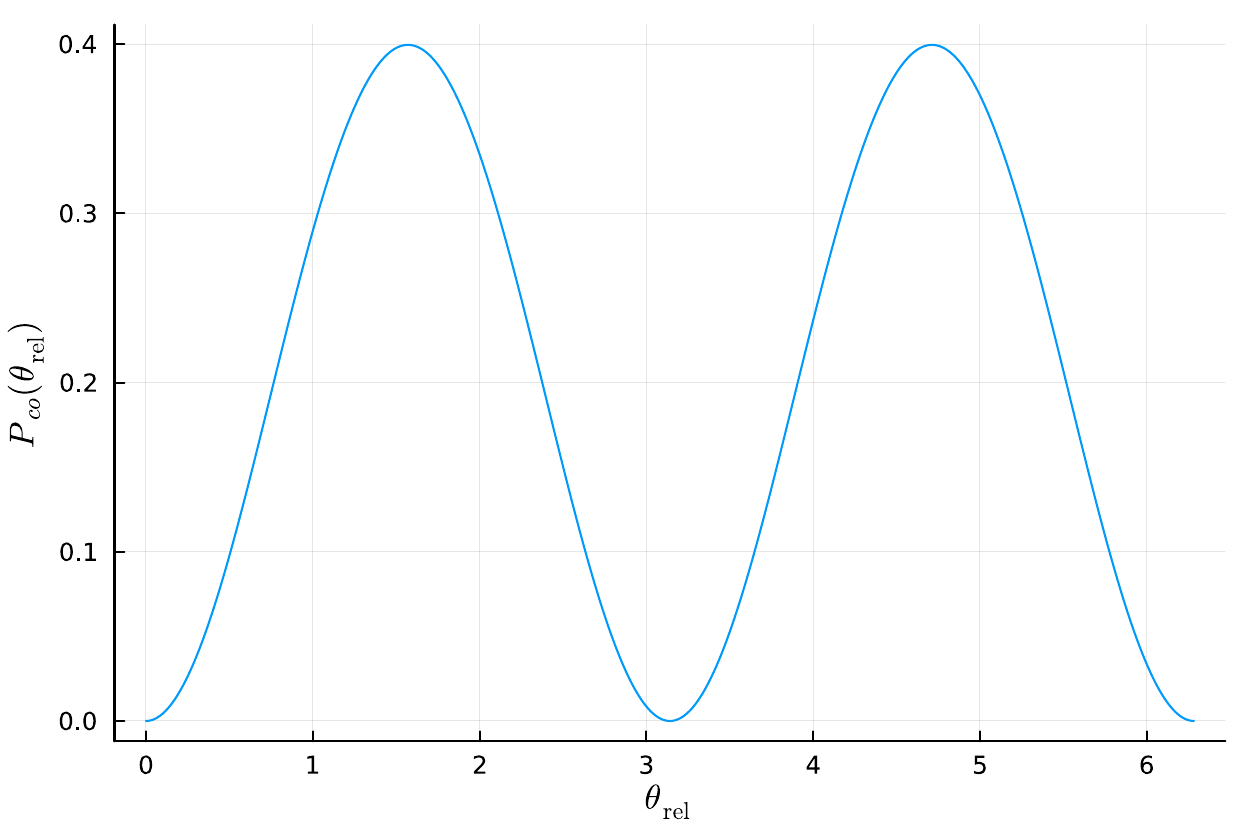}
    \caption{Coincidence probability in the Hong-Ou-Mandel set-up in quantum theory as a function of the initial relative phase between the two input states $\theta_\mathrm{rel}$. The input states are coherent states with parameters $\alpha_1^\mathrm{ini}=1$ and $ \alpha_2^\mathrm{ini}= \e^{i\theta_\mathrm{rel}}$.}
    \label{fig:quantum-coincidence}
\end{figure}

\begin{figure}
    \centering
    \includegraphics[width=\linewidth]{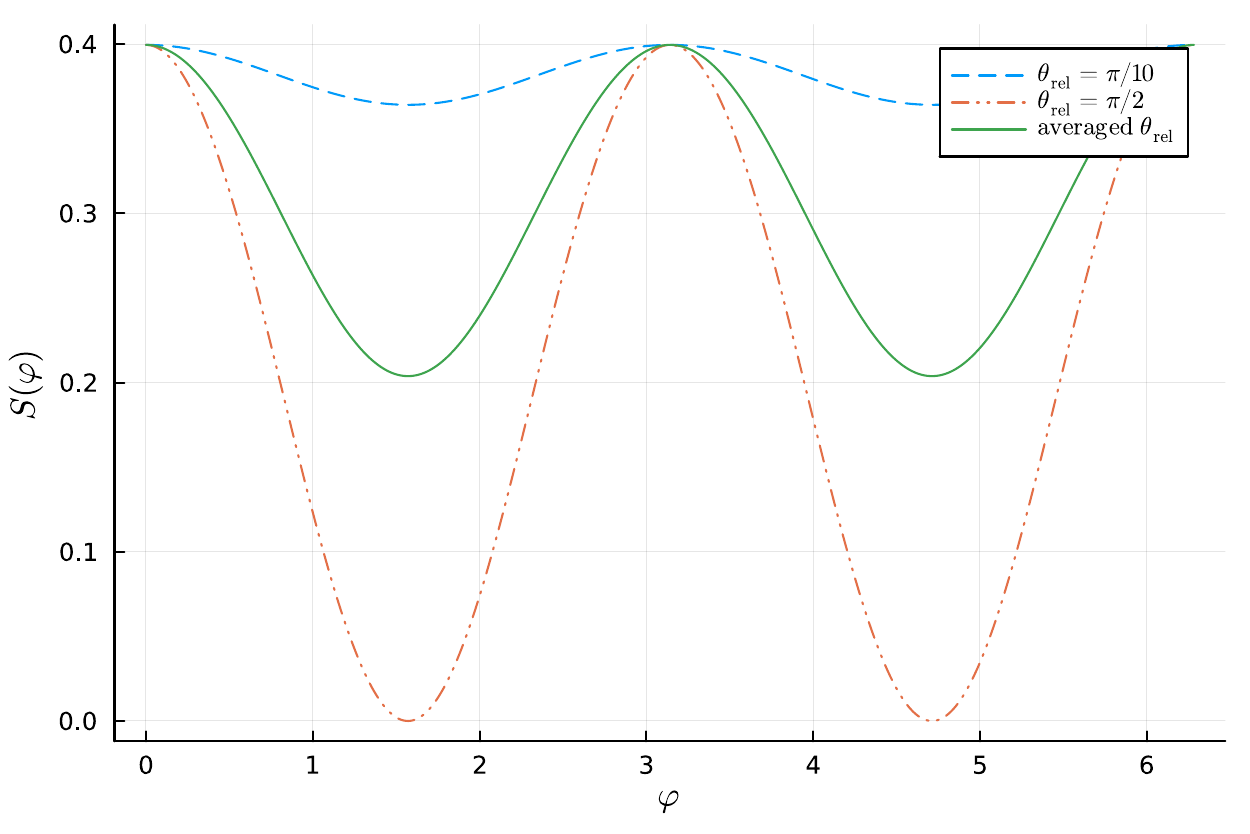}
    \caption{Measurement signal $S(\varphi)$ in the Mach-Zehnder interferometer in quantum theory as a function of the phase shift $\varphi$ introduced in one arm of the interferometer for different initial relative phases $\theta_\mathrm{rel}=\pi/10$ (dashed blue line), $\theta_\mathrm{rel}=\pi/2$ (red dotted-dashed line), and averaged over all relative phases (green solid line). The input states are coherent states with parameters $\alpha_1^\mathrm{ini}=1$ and $ \alpha_2^\mathrm{ini}= \e^{i\theta_\mathrm{rel}}$.}
    \label{fig:quantum-interference}
\end{figure}

The results can be generalised to coherent states instead of eigenstates of the Hamiltonian. We consider coherent states with mean energy equal to the energy of a single photon. The coherent state is characterised by a complex number $\alpha$ and thus the state is usually denoted as $\ket{\alpha}$. The real and imaginary parts of $\alpha$ correspond to the mean position and momentum of the coherent state. As we show next, phase shifters and beam splitters preserve the structure of coherent states and the coincidence measurements between two coherent states can also be expressed in terms of the respective complex numbers $\alpha_1$ and $\alpha_2$. The free time evolution of a coherent state corresponds to just changing the phase of $\alpha$: the time evolution is $\alpha(t) = \e^{-i \omega t} \alpha(0)$. Thus, unlike the eigenstates, the time evolution of coherent states is nontrivial, so any combination of input coherent states will have a well defined relative phase $\theta_\mathrm{rel}$.

Given two initial coherent states $\ket{\alpha_1^\mathrm{ini}}$ and $\ket{\alpha_2^\mathrm{ini}}$, using $\ket{\alpha} = \e^{-\abs{\alpha}^2/2} \e^{\alpha a^\dag} \ket{0}$, the output states after a general beam splitter are $\ket{\alpha^\mathrm{f}_1}$ and $\ket{\alpha^\mathrm{f}_2}$ with
\begin{equation}
    \begin{pmatrix}\alpha_1^\mathrm{f}\\\alpha_2^\mathrm{f}\end{pmatrix}
    = U^*
    \begin{pmatrix}\alpha_1^\mathrm{ini}\\\alpha_2^\mathrm{ini}\end{pmatrix}.
\end{equation}
It is important to note here that both phase shifters and beam splitters map two coherent states to two coherent states. In particular, the resulting state is still a product of coherent states.

For two coherent states, the coincidence probability is the product of probabilities to detect at least one photon
\begin{equation}
    P_\mathrm{co}(\ket{\alpha_1^\mathrm{f}} \ket{\alpha_2^\mathrm{f}}) = (1 - \abs{\braket{0|\alpha_1^\mathrm{f}}}^2) (1 - \abs{\braket{0|\alpha_2^\mathrm{f}}}^2).
\end{equation}
This, combined with $\braket{0|\alpha} = \e^{-\abs{\alpha}^2/2}$ yields
\begin{equation}
    P_\mathrm{co}(\ket{\alpha_1^\mathrm{f}} \ket{\alpha_2^\mathrm{f}}) = (1 - \e^{-\abs{\alpha_1^\mathrm{f}}^2})(1 - \e^{-\abs{\alpha_2^\mathrm{f}}^2}).
\end{equation}
In particular for two initial states with $\abs{\alpha_1^\mathrm{ini}}=\abs{\alpha_2^\mathrm{ini}}=1$, one has an initial coincidence probability of only $P_\mathrm{co}\approx 40\%$. Although it may be experimentally suitable to normalize all coincidence probabilities by this 40\%, we keep here the unnormalised form to get a result comparable to the other theories.

The coincidence probability as a function of the relative input phase is plotted in Fig.~\ref{fig:quantum-coincidence}. By integration over the relative phase, we obtain in this case a visibility of $V \approx 80\%$. In the interferometric set-up, we plot the output coincidence probability in Fig.~\ref{fig:quantum-interference}.

\subsection{Classical theory}
As already argued and experimentally demonstrated in Ref.~\cite{sadana2019near}, one can also observe the Hong-Ou-Mandel coincidence dip without requiring quantum optics if one is capable of controlling the relative phase of the incoming light. We will obtain the same result, and moreover show that in the Mach-Zehnder interferometer, classical theory also gives a measurement signal qualitatively similar to what we obtained from quantum theory.

\begin{figure}
    \centering
    \includegraphics[width=\linewidth]{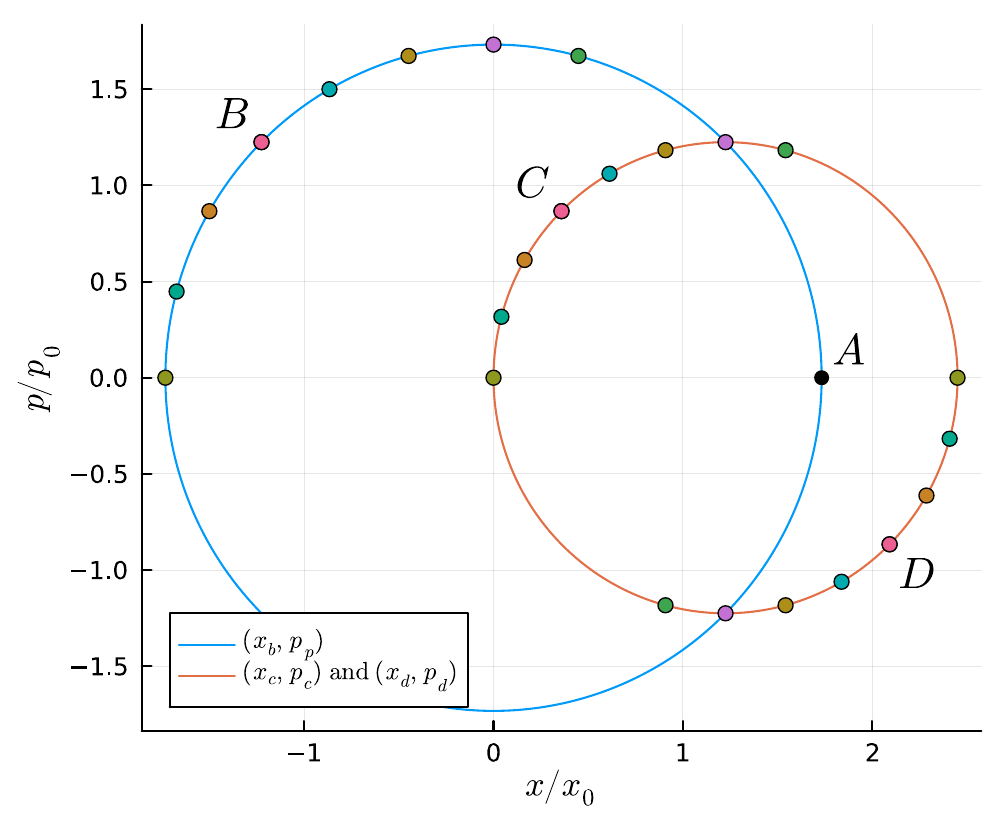}
    \caption{Points in phase space before and after passing through a beam splitter. Point $A$ represents the state of the first input oscillator which in our formulation is fixed. Points $B$ are the states of the second input oscillator for various values of $\theta_\mathrm{rel}$. Points $C$ and $D$ are the states of the output oscillators, these points are colour-coded so that same colour of points $B$, $C$, $D$ corresponds to the same relative phase $\theta_\mathrm{rel}$.}
    \label{fig:points}
\end{figure}

\begin{figure}
    \centering
    \includegraphics[width=\linewidth]{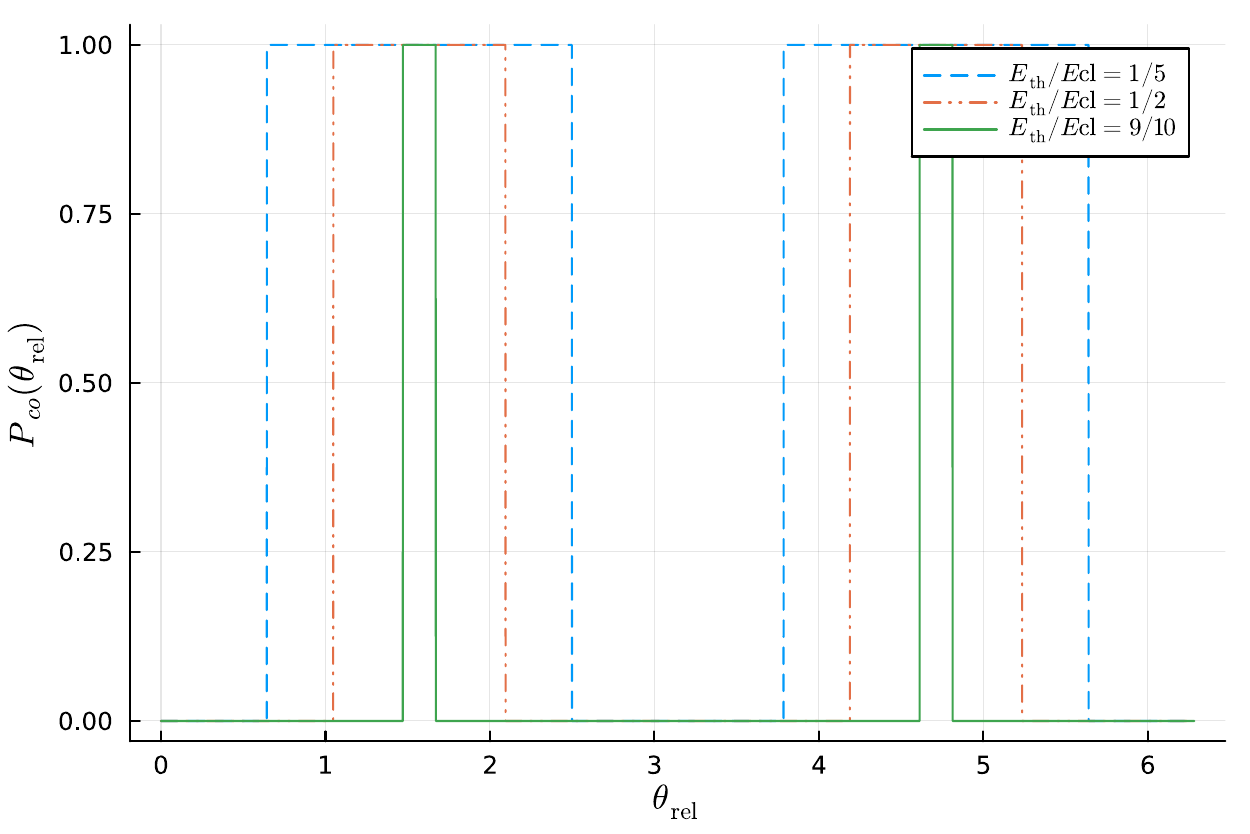}
    \caption{Probability of coincidence in classical theory in the Hong-Ou-Mandel set-up as a function of the initial relative phase $\theta_\mathrm{rel}$ for different energy thresholds $E_\mathrm{th}/E_\mathrm{cl}=1/5$ dashed blue line, $E_\mathrm{th}/E_\mathrm{cl}=1/2$ dashed-dotted red line, and $E_\mathrm{th}/E_\mathrm{cl}=9/10$ solid green line. The visibility of the Hong-Ou-Mandel coincidence dip (integral over $\theta_\mathrm{rel}$, see Eq.~\eqref{eqvisibility}) goes to $V\to 100\%$ as $E_\mathrm{th}\to E_\mathrm{cl}$.}
    \label{fig:classical-coincidence}
\end{figure}

\begin{figure}
    \centering
    \includegraphics[width=\linewidth]{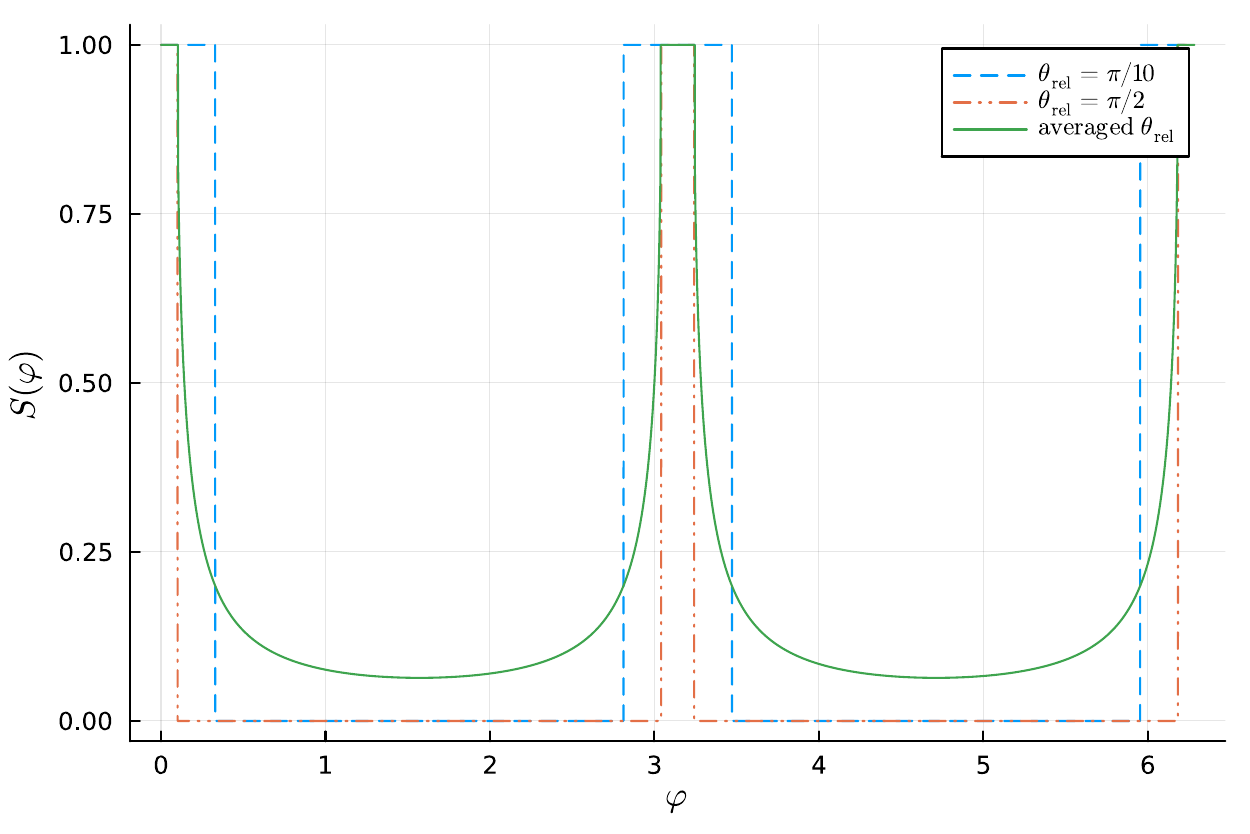}
    \caption{Output coincidence probability in Mach-Zehnder interferometer in classical theory as a function of the phase shift $\varphi$ introduced in one arm of the interferometer for different initial relative phases $\theta_\mathrm{rel}=\pi/10$ (dashed blue line), $\theta_\mathrm{rel}=\pi/2$ (red dotted-dashed line), and averaged over all initial relative phases (green solid line). Here we used $E_\mathrm{th}/E_\mathrm{cl}=9/10$.}
    \label{fig:classical-interference}
\end{figure}

In classical theory of electromagnetic radiation, the pure states correspond to Dirac delta distributions in phase space and the phase space spectral measure of the energy observable is given as:
\begin{equation}\label{eq:classicalspecfu}
    g_H(I; x,p) =
    \begin{cases}
    1 & H(x,p) \in I \\
    0 & H(x,p) \notin I.
    \end{cases}
\end{equation}
Classical theory does not have quantised energy levels, but the detection of a single photon can be mimicked by setting an arbitrary energy threshold $E_\mathrm{th}$, see Eq.~\eqref{probaco}. We take two classical harmonic oscillators (one per input port) with equal energy described by the joint state
\begin{multline}\label{eq:sharpstate}
    \rho_{AB}(x_1,p_1, x_2, p_2) = \\ \delta(x_1 - x_a) \delta(p_1 - p_a) \delta(x_2 - x_b) \delta(p_2 - p_b).
\end{multline}
By virtue of Lemma~\ref{lemma:BScommT}, a global phase shift commutes with the action of a beam splitter and, furthermore, the detection of the classical light is insensitive to a global phase. Hence we can set the phase of the first oscillator to zero. We can set the energy to an arbitrary value, but in order to match the energy of a single photon, we use $E_\mathrm{cl}=3 h_0/2$ and hence the initial coordinates read
\begin{equation}\label{eq:classinips}\begin{split}
    x_a &=r_\mathrm{cl} x_0, \\
    p_a &=0, \\
    x_b &=r_\mathrm{cl} x_0 \cos(\theta_\mathrm{rel}), \\
    p_b &=r_\mathrm{cl} p_0 \sin(\theta_\mathrm{rel}),
\end{split}\end{equation}
with $r_\mathrm{cl}=\sqrt 3$. After passing through the beam splitter, we get the transformed state
\begin{multline}
    \Phi_U(\rho_{AB}) (x_1,p_1, x_2, p_2) = \\ \delta(x_1 - x_c) \delta(p_1 - p_c) \delta(x_2 - x_d) \delta(p_2 - p_d),
\end{multline}
described by the coordinates transformed according to Eq.~\eqref{eqxpU}
\begin{equation}\label{statethetarel}\begin{split}
    x_c &= \frac{r_\mathrm{cl} x_0}{\sqrt{2}} (1 + \cos(\theta_\mathrm{rel})), \\
    p_c &= \frac{r_\mathrm{cl} p_0}{\sqrt{2}} \sin(\theta_\mathrm{rel}), \\
    x_d &= \frac{r_\mathrm{cl} x_0}{\sqrt{2}} (1 + \cos(\theta_\mathrm{rel} + \pi)), \\
    p_d &= \frac{r_\mathrm{cl} p_0}{\sqrt{2}} \sin(\theta_\mathrm{rel} + \pi),
\end{split}\end{equation}
see Fig.~\ref{fig:points} for a plot of the initial and final positions of these points as a function of $\theta_\mathrm{rel}$. Given the coordinates of these points, computing the coincidence probability of the transformed state $P_\mathrm{co}(\Phi_U(\rho_{AB}))$ finally requires setting the energy threshold $E_\mathrm{th}$. There is a clear upper bound as we need the input state to have $P_\mathrm{co}(\rho_{AB}) = 1$, so $E_\mathrm{th}<E_\mathrm{cl}$ must hold. Other than that, we are free to set the energy threshold $E_\mathrm{th}$ arbitrarily. In Fig.~\ref{fig:classical-coincidence} we plot the coincidence for different values of the energy threshold $E_\mathrm{th}$, showing that one can find suitable $\theta_\mathrm{rel}$ such that the coincidence probability is zero. Moreover we also see that as $E_\mathrm{th} $ approaches $E_\mathrm{cl}$, we get $P_\mathrm{co}(\Phi_U(\rho_{AB})) \approx 0$ for almost all values of $\theta_\mathrm{rel}$ and hence a visibility of $V\approx 100\%$. One can conclude the same from Fig.~\ref{fig:points}: for $E_\mathrm{th}/E_\mathrm{cl} = 1 - \varepsilon$ for small $\varepsilon > 0$ we have that almost any point inside the blue circle representing the initial states in Fig.~\ref{fig:points} has energy below $E_\mathrm{th}$. Hence, whenever one of the two outputs is inside the blue circle, the coincidence is zero. But this happens almost always since the output points are antipodal. Thus, even if we average over the initial relative phase $\theta_\mathrm{rel}$, the visibility approaches $V \to 100\% $ as $\varepsilon \to 0$.

Finally, we consider the Mach-Zehnder interferometer. In this case, using the already outlined methods, one can again compute the output coincidence probability. We do this in cases when the initial relative phase $\theta_\mathrm{rel}$ is set to a concrete value, but also when we average over all possible $\theta_\mathrm{rel}$. The results, see Fig.~\ref{fig:classical-interference}, show that classical theory produces a measurement signal that is only quantitatively but not qualitatively different from the one obtained from quantum theory given in Eq.~\eqref{eq:QT-measurement}.

\subsection{Sawtooth theory}\label{sec:sawtooth}
Sawtooth theory is a general phase-space theory, where the phase-space spectral measure of the Hamiltonian has a sawtooth (or triangular) shape, and the energy spectrum is $E_n=h_0\,n/2$ with $n=0,1,2,\dotsc$. This theory allows for rather classical or rather quantum states \cite{plavala2022operational}, but we will focus on states without preparation uncertainty, that is, Dirac delta distributions in phase space as in Eq.~\eqref{eq:classicalspecfu}. Note that the sawtooth theory as defined in Ref.~\cite{plavala2022operational} has an energy spacing of $h_0/2$ and the first excited state has energy $h_0/2$, hence being at variance with quantum theory.

To calculate the coincidence probability $P_\mathrm{co}$, we set the energy threshold halfway between $E_0$ and $E_1$, that is, $E_\mathrm{th}=h_0/4$. According to Ref.~\cite{plavala2022operational}, the spectral measure yields
\begin{equation}
    g_H((E_\mathrm{th},\infty); r) = 1-T_0(r^2)= \min(r^2,1).
\end{equation}
We use the same states as in classical theory, see Eq.~\eqref{eq:sharpstate} and Eq.~\eqref{eq:classinips}, but with $r_\mathrm{cl}$ replaced by $1$ corresponding to the first excited state with $H=h_0/2$. Then we have
\begin{equation}
    P_\mathrm{co}(\Phi_U (\rho_{AB})) = \min(r_c^2,1)\min(r_d^2,1),
\end{equation}
where $r_c$ and $r_d$ are given by Eq.~\eqref{statethetarel} (again with $r_\mathrm{cl}$ replaced by $1$). For the relative phase $\theta_\mathrm{rel}=0$ ($\theta_\mathrm{rel}=\pi$) we have $r_d=0$ ($r_c=0$) which implies that the coincidence probability is zero.

However, integrating over all relative phases provides a nonzero coincidence probability in general, as can be seen from Fig.~\ref{fig:sawtooth-coincidence}, where the coincidence probability $P_\mathrm{co}=1-\abs{\cos(\theta_\mathrm{rel})}$ is plotted as a function of the relative phase between the two inputs $\theta_\mathrm{rel}$. We see that contrary to classical theory, the coincidence probability does not vanish when integrating the phase. Instead, it reaches a maximum visibility of $V = 2/\pi \approx 64\%$. Note however that this maximum is only valid for the specific states we considered: in principle, sawtooth theory can accommodate states that are negative in some phase-space regions \cite{plavala2022operational} that may allow reaching higher visibilities, but in turn, are further away from classical theory.

\begin{figure}
    \centering
    \includegraphics[width=\linewidth]{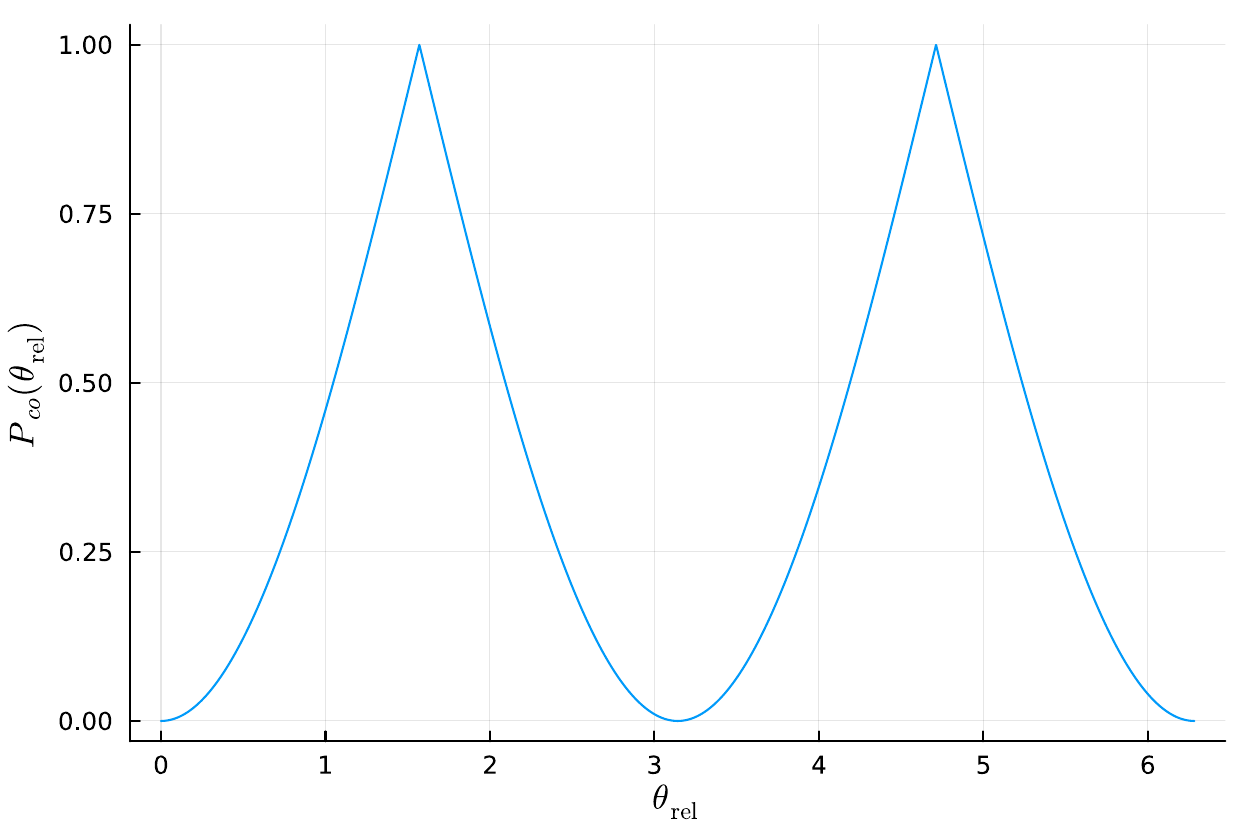}
    \caption{Output coincidence probability in the Hong-Ou-Mandel set-up in the sawtooth theory as a function of the initial relative phase $\theta_\mathrm{rel}$ between the two input states.}
    \label{fig:sawtooth-coincidence}
\end{figure}

We also compute the measurement signal $S(\varphi)=1-\abs{\sin(\theta_\mathrm{rel})\sin(\varphi)}$ in the Mach-Zehnder set-up in the sawtooth theory; the results are depicted in Fig.~\ref{fig:sawtooth-interference}. This plot shows that we once again obtain an interference pattern, especially if we are able to control the initial relative phase $\theta_\mathrm{rel}$.

\begin{figure}
    \centering
    \includegraphics[width=\linewidth]{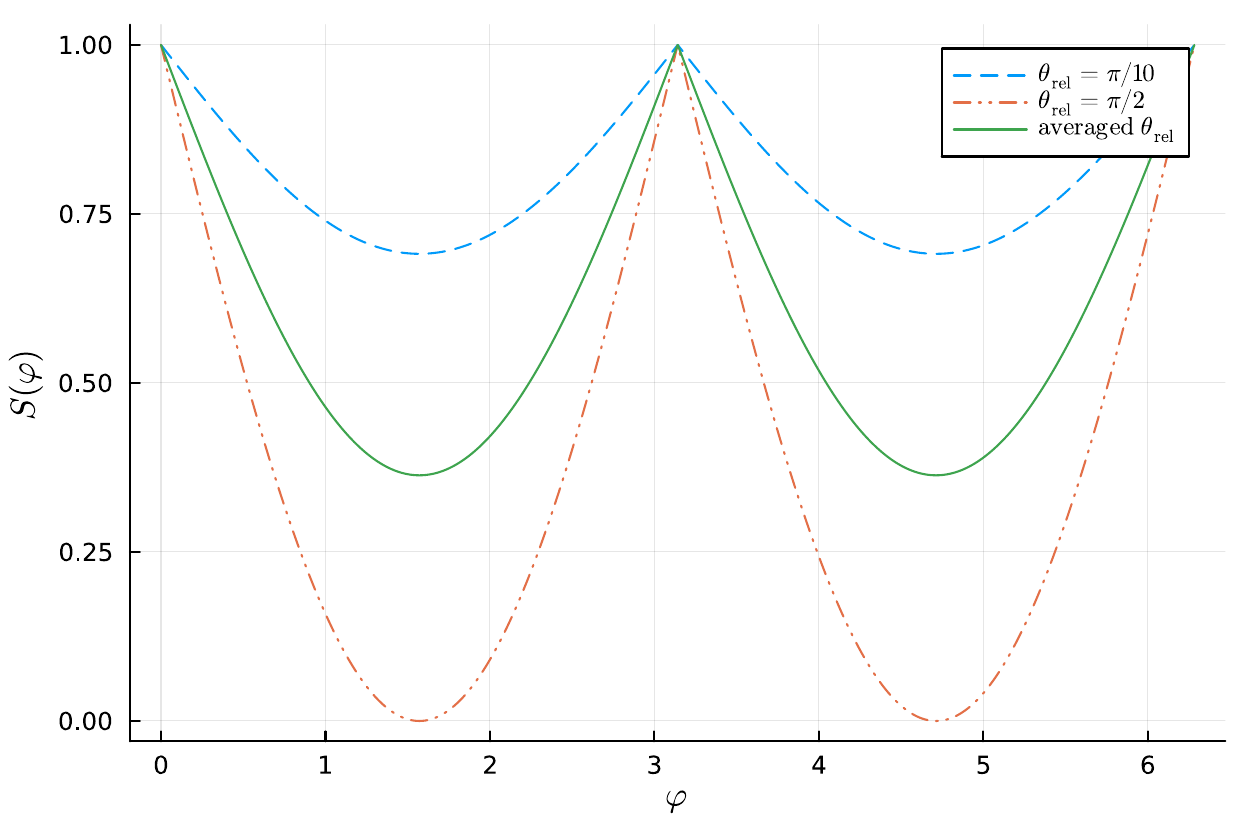}
    \caption{Output coincidence probability in the Mach-Zehnder interferometer in the sawtooth theory as a function of the phase shift $\varphi$ introduced in one arm of the interferometer for different initial relative phases $\theta_\mathrm{rel}=\pi/10$ (dashed blue line), $\theta_\mathrm{rel}=\pi/2$ (red dotted-dashed line), and averaged over all initial relative phases (green solid line).}
    \label{fig:sawtooth-interference}
\end{figure}


\section{Correlations in an interferometric set-up}
The previous section showed that the difference between the coincidence probabilities in the two set-ups (Hong-Ou-Mandel and Mach-Zehnder) obtained from the three considered theories is quantitative rather than qualitative. This raises the question of which phenomena, if any, are qualitatively different for quantum optics compared to other theories. Neither the Hong-Ou-Mandel effect nor the interferometric configuration is sufficient to address this question. It was suggested in Refs.~\cite{catani2023interference, hance2026noncontextual} that contextuality provides a more stringent criterion for nonclassicality than interference alone. We investigate here a simple set-up built on top of the Mach-Zehnder interferometer that enables quantum theory to violate contextuality inequalities derived from the Clauser-Horne-Shimony-Holt inequality in Ref.~\cite{plavala2024contextuality}. Violations of such contextuality inequalities exclude the possibility of certain classical hidden-variable models. Here we describe a set-up that is in principle suitable for testing contextuality inequalities and evaluate it in quantum, classical, and sawtooth theory.

Specifically, we consider the sequence of beam splitter, phase shifter, and beam splitter of a Mach-Zehnder interferometer, see Fig.~\ref{figHOM_interferometer}, as a transformation $\Psi_\varphi$ parameterised by the phase shift $\varphi$. This transformation can be applied in classical, quantum, and sawtooth theory. Moreover, we will consider two distinct measurements: the first one is the coincidence measurement we considered before; the second measurement is the ``Hadamard'' measurement that is the coincidence measurement preceded by an additional interferometer with $\varphi=\pi/4$, yielding an additional transformation $\Psi_{\pi / 4}$. As we will shortly see, this measurement corresponds to the measurement in the Hadamard basis in quantum theory.

Given a theory and these two measurements, we consider a set of input states, transform them using $\Psi_\varphi$ for all $\varphi$, and then measure with either of the two coincidence measurements. We then record the joint numerical range of both measurements, that is, to each state we assign coordinates $(x,y)$ based on the two coincidence probabilities of the measurements. For this, we calculate the coincidence probabilities as we vary the phase shift $\varphi$ in the interferometer and traverse through the set of initial states. These probability pairs form regions containing all measurement statistics accessible in the theory. Joint numerical ranges can be also understood as effective state spaces representing the effective generalised probabilistic theories \cite{plavala2023general} we obtained in this way. Differences between theories manifest as different shapes of the accessible state space, as we shall see in the following.

In quantum theory the Mach-Zehnder interferometer produces from $\ket{11}$ a superposition of
\begin{equation}
    \ket{a} = \ket{11} \quad\text{and}\quad
    \ket{b} = -\frac{i}{\sqrt2}\left(\ket{20} + \ket{02}\right)
\end{equation}
Hence $\Psi_\varphi$ implements a unitary transformation $U_\varphi$ in the subspace spanned by $\ket a$ and $\ket b$, described in the $\ket{a}$, $\ket{b}$ basis by the matrix
\begin{equation}
    U_\varphi =
    \begin{pmatrix}
    \cos(\varphi ) & \sin(\varphi ) \\
    -\sin(\varphi ) & \cos(\varphi )
    \end{pmatrix}.
\end{equation}
It was shown in Ref.~\cite{plavala2024contextuality} that a single qubit and measurements in two distinct bases (related by a $\pi/4$ rotation) are sufficient to violate a contextuality inequality. The two measurements we consider consist of two Mach-Zehnder interferometers corresponding to the unitary matrices $U_{\varphi=0}$ and $U_{\varphi=\pi/4}$, respectively.

\begin{figure}
    \centering
    \includegraphics[width=\linewidth]{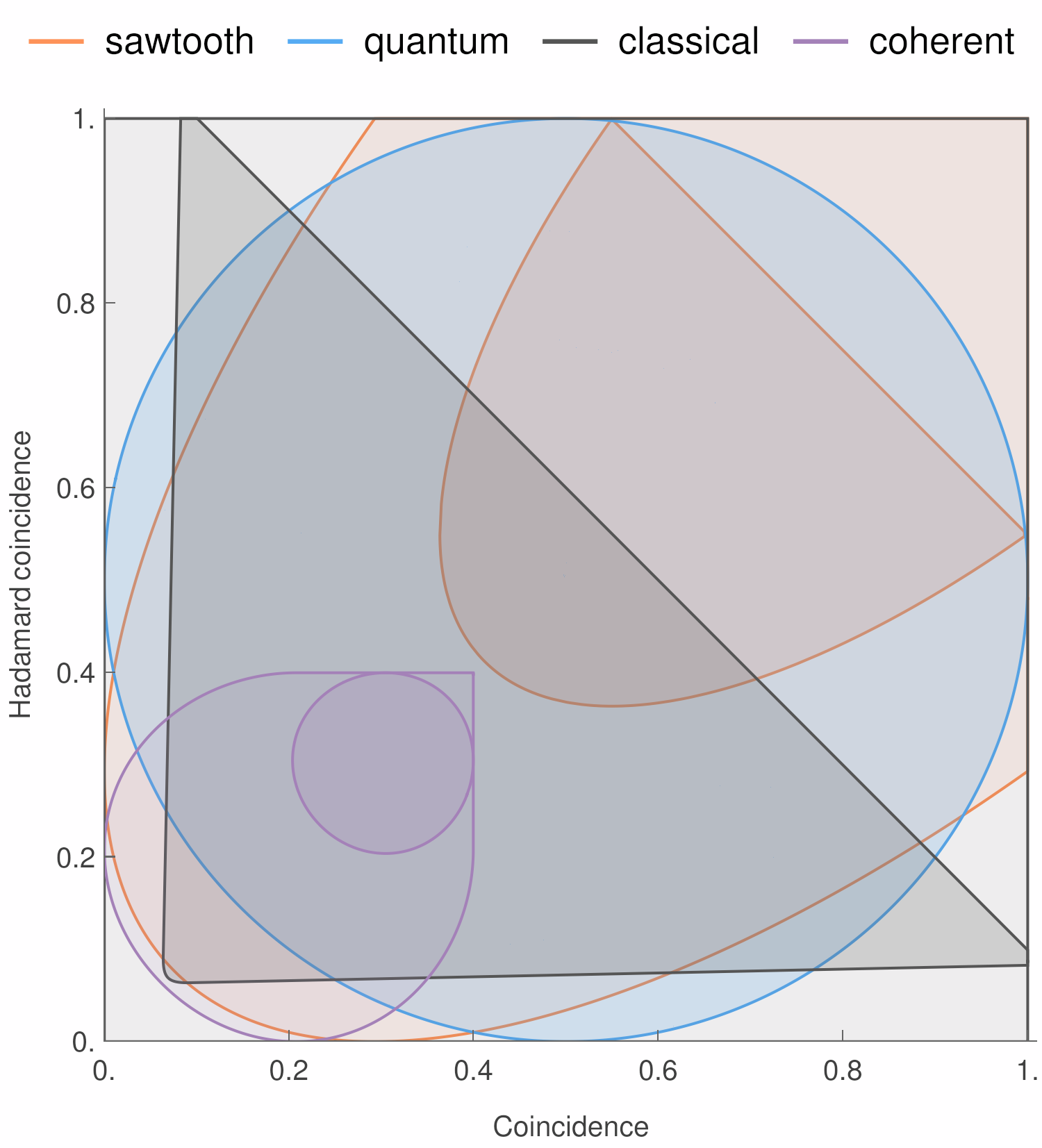}
    \caption{Joint numerical ranges for the Mach-Zehnder interferometer: After a Mach-Zehnder interferometer with phase shift $\varphi$, a coincidence measurement (horizontal axis) or a Hadamard coincidence measurement (vertical axis) is performed. A point $(x,y)$ is in a numerical range, if there is an input state and a phase shift giving a coincidence probability $x$ and a Hadamard coincidence probability $y$. In quantum theory, with initial state $\ket{11}$, we obtain a circle (blue). The other depicted sets are classical theory (black square and triangle), sawtooth theory (brown parabolic caps), and coherent quantum states (small purple shapes). The larger of each shape occurs when the initial relative phase $\theta_\mathrm{rel}$ can be controlled, the smaller shape when the average over $\theta_\mathrm{rel}$ is taken. In classical theory, the energy threshold $E_\mathrm{th}/E_\mathrm{cl} = 9/10$ is used.}
    \label{fig:effective-state-spaces}
\end{figure}

In both classical and sawtooth theory, we consider the same set of initial states as before, parameterised by the relative phase $\theta_\mathrm{rel}$ (see Eq.~\eqref{eq:sharpstate}, Eq.~\eqref{eq:classinips} and Sec.~\ref{sec:sawtooth}), as well as only the initial state averaged over the relative phase. We also consider coherent states in quantum mechanics as before. The results are shown in Fig.~\ref{fig:effective-state-spaces}. For quantum theory with initial state $\ket{11}$, the joint numerical range is a circle. This is not surprising, as this coincides with the corresponding section of the Bloch sphere. In classical theory with control of the initial relative phase $\theta_\mathrm{rel}$, the joint numerical range of the two measurements is the whole possible region of values; that is, any possible combination of probabilities for the two measurements can be reached, implying that the two measurements are independent in classical theory. This means that classical theory contains the numerical range of quantum theory and that classical theory can maximally violate any contextuality inequality based on these two measurements. We stress that this violation by a classical theory comes from the fact that the contextuality inequality on Ref.~\cite{plavala2024contextuality} is more restrictive and relies on the assumption that the two measurements can distinguish all values of the parameter $\varphi$, which is not satisfied in the specific classical theory used here. Nevertheless, this observation highlights that at the level of this restricted interferometric set-up, classical theory can reproduce the range of correlations accessible to quantum theory. Qualitative distinctions between the theories only emerge once additional structural constraints are imposed. Finally, sawtooth theory and quantum coherent states lead to predictions that are different from quantum theory but also compatible with classical theory.


\section{Conclusions}
We developed a model to study linear optics in the framework of generalised probability theories in phase space. We studied beam splitters and phase shifters in detail and investigated the boundary between classical and quantum optics by discussing the visibility of the Hong-Ou-Mandel coincidence dip and Mach-Zehnder interferometry in interferometric set-ups. Our work can be used as a basis to study more complicated photonic circuits represented. This will enable drawing a better boundary between classical and quantum photonic computers, either boson samplers or universal quantum computers. We found that in our framework, all of the theories, including classical and quantum theory, produce qualitatively similar predictions. Since arguably photonic quantum computing is based on similar ideas, our work suggests that drawing a clear distinction between quantum photonic computer and analogue classical optical computer is not trivial, and that the same is true for generalized theories, that are neither classical nor quantum.

\begin{acknowledgments}
This work was supported by the Deutsche Forschungsgemeinschaft (DFG, German Research Foundation, project numbers 447948357, 440958198, 563437167, and 572827488), the Sino-German Center for Research Promotion (Project M-0294), and the German Federal Ministry of Research, Technology and Space (Project QuKuK, Grant No.~16KIS1618K and Project BeRyQC, Grant No.~13N17292). MP acknowledges support from the Niedersächsisches Ministerium für Wissenschaft und Kultur.
\end{acknowledgments}

\bibliography{biblio}

@Article{plavala2023general,
  author  = {Pl{\'a}vala, Martin},
  journal = {Physics Reports},
  title   = {General probabilistic theories: An introduction},
  year    = {2023},
  pages   = {1--64},
  volume  = {1033},
  doi     = {10.1016/j.physrep.2023.09.001},
}

@Article{knill2001scheme,
  author  = {Knill, Emanuel and Laflamme, Raymond and Milburn, Gerald J.},
  journal = {Nature},
  title   = {A scheme for efficient quantum computation with linear optics},
  year    = {2001},
  number  = {6816},
  pages   = {46--52},
  volume  = {409},
  doi     = {10.1038/35051009},
}

@Article{bouchard2020two,
  author  = {Bouchard, Fr{\'e}d{\'e}ric and Sit, Alicia and Zhang, Yingwen and Fickler, Robert and Miatto, Filippo M. and Yao, Yuan and Sciarrino, Fabio and Karimi, Ebrahim},
  journal = {Reports on Progress in Physics},
  title   = {Two-photon interference: the {H}ong--{O}u--{M}andel effect},
  year    = {2020},
  number  = {1},
  pages   = {012402},
  volume  = {84},
  doi     = {10.1088/1361-6633/abcd7a},
}

@Article{sadana2019near,
  author  = {Sadana, Simanraj and Ghosh, Debadrita and Joarder, Kaushik and Lakshmi, A. Naga and Sanders, Barry C. and Sinha, Urbasi},
  journal = {Physical Review A},
  title   = {Near-100\% two-photon-like coincidence-visibility dip with classical light and the role of complementarity},
  year    = {2019},
  number  = {1},
  pages   = {013839},
  volume  = {100},
  doi     = {10.1103/physreva.100.013839},
}

@Article{plavala2023generalized,
  author  = {Pl{\'a}vala, Martin and Kleinmann, Matthias},
  journal = {Physical Review A},
  title   = {Generalized dynamical theories in phase space and the hydrogen atom},
  year    = {2023},
  number  = {5},
  pages   = {052212},
  volume  = {108},
  doi     = {10.1103/physreva.108.052212},
}

@Article{plavala2022operational,
  author  = {Pl{\'a}vala, Martin and Kleinmann, Matthias},
  journal = {Physical Review Letters},
  title   = {Operational Theories in Phase Space: Toy Model for the Harmonic Oscillator},
  year    = {2022},
  number  = {4},
  pages   = {040405},
  volume  = {128},
  doi     = {10.1103/physrevlett.128.040405},
}

@Article{plavala2024contextuality,
  author  = {Pl{\'a}vala, Martin and G{\"u}hne, Otfried},
  journal = {Physical Review Letters},
  title   = {Contextuality as a Precondition for Quantum Entanglement},
  year    = {2024},
  number  = {10},
  pages   = {100201},
  volume  = {132},
  doi     = {10.1103/physrevlett.132.100201},
}

@Article{brunner2014bell,
  author  = {Brunner, Nicolas and Cavalcanti, Daniel and Pironio, Stefano and Scarani, Valerio and Wehner, Stephanie},
  journal = {Reviews of Modern Physics},
  title   = {Bell nonlocality},
  year    = {2014},
  number  = {2},
  pages   = {419--478},
  volume  = {86},
  doi     = {10.1103/revmodphys.86.419},
}

@Article{case2008wigner,
  author  = {Case, William B.},
  journal = {American Journal of Physics},
  title   = {{W}igner functions and {W}eyl transforms for pedestrians},
  year    = {2008},
  number  = {10},
  pages   = {937--946},
  volume  = {76},
  doi     = {10.1119/1.2957889},
}

@Article{zhong2020quantum,
  author  = {Zhong, Han-Sen and Wang, Hui and Deng, Yu-Hao and Chen, Ming-Cheng and Peng, Li-Chao and Luo, Yi-Han and Qin, Jian and Wu, Dian and Ding, Xing and Hu, Yi and Hu, Peng and Yang, Xiao-Yan and Zhang, Wei-Jun and Li, Hao and Li, Yuxuan and Jiang, Xiao and Gan, Lin and Yang, Guangwen and You, Lixing and Wang, Zhen and Li, Li and Liu, Nai-Le and Lu, Chao-Yang and Pan, Jian-Wei},
  journal = {Science},
  title   = {Quantum computational advantage using photons},
  year    = {2020},
  number  = {6523},
  pages   = {1460--1463},
  volume  = {370},
  doi     = {10.1126/science.abe8770},
}

@InProceedings{aaronson2011computational,
  author     = {Aaronson, Scott and Arkhipov, Alex},
  booktitle  = {Proceedings of the Forty-Third Annual ACM Symposium on Theory of Computing},
  title      = {The computational complexity of linear optics},
  year       = {2011},
  pages      = {333--342},
  publisher  = {ACM},
  series     = {STOC'11},
  collection = {STOC'11},
  doi        = {10.1145/1993636.1993682},
}

@Article{catani2023interference,
  author  = {Catani, Lorenzo and Leifer, Matthew and Schmid, David and Spekkens, Robert W.},
  journal = {Quantum},
  title   = {Why interference phenomena do not capture the essence of quantum theory},
  year    = {2023},
  pages   = {1119},
  volume  = {7},
  doi     = {10.22331/q-2023-09-25-1119},
}

@Article{na2020classical,
  author  = {Na, Dong-Yeop and Chew, Weng Cho},
  journal = {Progress In Electromagnetics Research},
  title   = {Classical and quantum electromagnetic interferences: What is the difference?},
  year    = {2020},
  pages   = {1--13},
  volume  = {168},
  doi     = {10.2528/pier20060301},
}

@Article{aspect1981experimental,
  author  = {Aspect, Alain and Grangier, Philippe and Roger, G{\'e}rard},
  journal = {Physical Review Letters},
  title   = {Experimental Tests of Realistic Local Theories via {B}ell's Theorem},
  year    = {1981},
  number  = {7},
  pages   = {460},
  volume  = {47},
  doi     = {10.1103/physrevlett.47.460},
}

@Article{aspect1982experimental,
  author  = {Aspect, Alain and Dalibard, Jean and Roger, G{\'e}rard},
  journal = {Physical Review Letters},
  title   = {Experimental test of {B}ell's inequalities using time-varying analyzers},
  year    = {1982},
  number  = {25},
  pages   = {1804},
  volume  = {49},
  doi     = {10.1103/physrevlett.49.1804},
}

@Article{aspect1982experimental2,
  author  = {Aspect, Alain and Grangier, Philippe and Roger, G{\'e}rard},
  journal = {Physical Review Letters},
  title   = {Experimental Realization of {E}instein-{P}odolsky-{R}osen-{B}ohm Gedankenexperiment: A New Violation of {B}ell's Inequalities},
  year    = {1982},
  number  = {2},
  pages   = {91},
  volume  = {49},
  doi     = {10.1103/physrevlett.49.91},
}

@Article{hong1987measurement,
  author  = {Hong, C.-K. and Ou, Z.-Y. and Mandel, L.},
  journal = {Physical Review Letters},
  title   = {Measurement of subpicosecond time intervals between two photons by interference},
  year    = {1987},
  number  = {18},
  pages   = {2044},
  volume  = {59},
  doi     = {10.1103/physrevlett.59.2044},
}

@Article{shih1988new,
  author  = {Shih, Y. H. and Alley, Carroll O.},
  journal = {Physical Review Letters},
  title   = {New type of {E}instein-{P}odolsky-{R}osen-{B}ohm experiment using pairs of light quanta produced by optical parametric down conversion},
  year    = {1988},
  number  = {26},
  pages   = {2921},
  volume  = {61},
  doi     = {10.1103/physrevlett.61.2921},
}

@Book{grynberg2010introduction,
  author    = {Grynberg, Gilbert and Aspect, Alain and Fabre, Claude and Cohen-Tannoudji, Claude},
  publisher = {Cambridge University Press},
  title     = {Introduction to Quantum Optics: From the Semi-classical Approach to Quantized Light},
  year      = {2010},
  doi       = {10.1017/cbo9780511778261},
}

@Article{t2024hidden,
  author  = {'t Hooft, Gerard},
  journal = {Frontiers in Quantum Science and Technology},
  title   = {The hidden ontological variable in quantum harmonic oscillators},
  year    = {2024},
  pages   = {1505593},
  volume  = {3},
  doi     = {10.3389/frqst.2024.1505593},
}

@Unpublished{hance2026noncontextual,
  author        = {Hance, Jonte R. and Krnic, Jakov and Larsson, Jan-\AA{}ke},
  title         = {Noncontextual versus contextual interferometry},
  year          = {2026},
  archiveprefix = {arXiv},
  doi           = {10.48550/arXiv.2601.13109},
  eprint        = {2601.13109},
  eprinttype    = {arxiv},
  primaryclass  = {quant-ph},
}

@Unpublished{wetterich2025quantum,
  author        = {Wetterich, Christof},
  title         = {Quantum evolution with classical fields},
  year          = {2025},
  archiveprefix = {arXiv},
  doi           = {10.48550/arXiv.2510.24275},
  eprint        = {2510.24275},
  eprinttype    = {arxiv},
  primaryclass  = {quant-ph},
}

@article{Kwiat_1995,
 author = {Kwiat, Paul and Weinfurter, Harald and Herzog, Thomas and Zeilinger, Anton and Kasevich, Mark A.},
 doi = {10.1103/physrevlett.74.4763},
 journal = {Physical Review Letters},
 number = {24},
 pages = {4763},
 publisher = {American Physical Society (APS)},
 title = {Interaction-Free Measurement},
 volume = {74},
 year = {1995}
}

@article{Elitzur_1993,
 author = {Elitzur, Avshalom C. and Vaidman, Lev},
 doi = {10.1007/bf00736012},
 journal = {Foundations of Physics},
 number = {7},
 pages = {987},
 publisher = {Springer Science and Business Media LLC},
 title = {Quantum mechanical interaction-free measurements},
 volume = {23},
 year = {1993}
}

@article{Degen_2017,
 author = {Degen, C. L. and Reinhard, F. and Cappellaro, P.},
 doi = {10.1103/revmodphys.89.035002},
 journal = {Reviews of Modern Physics},
 number = {3},
 publisher = {American Physical Society (APS)},
 title = {Quantum sensing},
 url = {http://dx.doi.org/10.1103/revmodphys.89.035002},
 volume = {89},
 year = {2017},
 pages = {035002},
}
\onecolumngrid

\appendix


\section{Time evolution}\label{appA}
We use $\Phi_t$ to denote the super-operator that shifts the system in time by $t$, that is $\Phi_t(\rho)$ is the state $\rho$ of a system shifted forward by $t$ in time. For general phase-space theories, a generalisation of the Moyal bracket was suggested \cite{plavala2023generalized}, which, conveniently, for the case of harmonic Hamiltonians reduces to the Liouville equation \cite{plavala2023generalized}. Hence we have for harmonic Hamiltonians, $\dot \Phi_t(\rho) = \lbrace H, \Phi_t(\rho) \rbrace$ with the Poisson bracket $\{f,g\} = \frac{\partial f}{\partial x} \frac{\partial g}{\partial p} - \frac{\partial f}{\partial p} \frac{\partial g}{\partial x}$. For the Hamiltonian of the harmonic oscillator, the solution of the Liouville equation is
\begin{equation}
    \Phi_t(\rho)(x,p)= \rho\big(
    x \cos(\omega t)-\frac{x_0}{p_0}p \sin(\omega t),
    p \cos(\omega t)+\frac{p_0}{x_0}x \sin(\omega t) \big).
\end{equation}
Hence, when writing $\rho$ in polar coordinates, $x=x_0 r \cos(\theta)$ and $p=p_0 r \sin(\theta)$, we obtain
$\Phi_t(\rho)(r,\theta)= \rho(r,\theta+\omega t)$, as used in Eq.~\eqref{eq:timevo} of the main text.

\end{document}